\def\lipicsf{0}
\ifnum\lipicsf=1
\documentclass[a4paper, UKenglish]{lipics-v2018}
\pdfoutput=1
\usepackage{algorithm}
\usepackage{algorithmic}
\usepackage{comment}
\usepackage{color,soul}
\usepackage{wrapfig}
\usepackage{tikz}
\usetikzlibrary{shapes,arrows,automata}
\usetikzlibrary{decorations.pathreplacing}
\usetikzlibrary{positioning}
\tikzset{
          box/.style={rectangle,draw=black,thick, minimum size=1cm},
}
\usepackage{pgfplots}
\title{Asynchronous Filling by Myopic Luminous Robots\footnote{The research has been partially supported by the European Union, co-financed by the European Social Fund (EFOP-3.6.3-VEKOP-16-2017-00002).}} %

\titlerunning{Asynchronous Filling by Myopic Luminous Robots}%

\author{Attila Hideg}{Department of Automation and Applied Informatics, Budapest University of Technology and Economics, Budapest, Hungary}{Attila.Hideg@aut.bme.hu}{}{}
\author{Tam\'as Lukovszki}{Faculty of Informatics E\"{o}tv\"{o}s Lor\'and University, Budapest, Hungary}{lukovszki@inf.elte.hu}{}{}

\authorrunning{A. Hideg and T. Lukovszki}%

\Copyright{Attila Hideg and Tam\'as Lukovszki}%

\ccsdesc[100]{Theory of computation~Design and analysis of algorithms} %

\keywords{Autonomous robots, Asynchronous model, Dispersing, Filling} %

\nolinenumbers %

\hideLIPIcs  %

\EventEditors{}
\EventNoEds{1}
\EventLongTitle{}
\EventShortTitle{}
\EventAcronym{}
\EventYear{}
\EventDate{}
\EventLocation{}
\EventLogo{}
\SeriesVolume{}
\ArticleNo{1}
\else %
\documentclass[runningheads]{llncs}
\pdfoutput=1

	\usepackage{algorithm}
	\usepackage{algorithmic}
	\usepackage{comment}
  	\usepackage{color,soul}
	\usepackage{wrapfig}
	\usepackage{tikz}
    \usetikzlibrary{shapes,arrows,automata}
    \usetikzlibrary{decorations.pathreplacing}
    \usetikzlibrary{positioning}
    \tikzset{
          box/.style={rectangle,draw=black,thick, minimum size=1cm},
    }
    \usepackage{pgfplots}
	\title{Asynchronous Filling by Myopic Luminous Robots\thanks{The research has been partially supported by the European Union, co-financed by the European Social Fund (EFOP-3.6.3-VEKOP-16-2017-00002).}} %
	\titlerunning{Asynchronous Filling by Myopic Luminous Robots}
	\author{Attila Hideg\inst{1} \and
			Tam\'as Lukovszki\inst{2}}
	\authorrunning{A. Hideg and T. Lukovszki}
	\institute{Department of Automation and Applied Informatics,\\
			   Budapest University of Technology and Economics, Budapest, Hungary
					\email{Attila.Hideg@aut.bme.hu}\and
		   	   Faculty of Informatics, E\"otv\"os Lor\'and University, Budapest, Hungary\\
					\email{lukovszki@inf.elte.hu}}

\fi

\makeatletter
\newenvironment{breakablealgorithm}
  {%
   \begin{center}
     \refstepcounter{algorithm}%
     \hrule height.8pt depth0pt \kern2pt%
     \renewcommand{\caption}[2][\relax]{%
       {\raggedright\textbf{\ALG@name~\thealgorithm} ##2\par}%
       \ifx\relax##1\relax %
         \addcontentsline{loa}{algorithm}{\protect\numberline{\thealgorithm}##2}%
       \else %
         \addcontentsline{loa}{algorithm}{\protect\numberline{\thealgorithm}##1}%
       \fi
       \kern2pt\hrule\kern2pt
     }
  }{%
     \kern2pt\hrule\relax%
   \end{center}
  }

\begin{document}

\maketitle

\begin{abstract} 
We consider the problem of filling an unknown area represented by an arbitrary connected graph of $n$ vertices by mobile luminous robots. In this problem, the robots enter the graph one-by-one through a specific vertex, called the  Door, and they eventually have to cover all vertices of the graph while avoiding collisions. The robots are anonymous and make decisions driven by the same local rule of behavior. They have limited persistent memory and limited visibility range. We investigate the Filling problem in the asynchronous model.

We assume that the robots know an upper bound $\Delta$ on the maximum degree of the graph before entering.  
We present an algorithm solving the asynchronous Filling problem with robots having $1$ hop visibility range, $O(\log\Delta)$ bits of persistent storage, and $\Delta+4$ colors, including the color when the light is off. We analyze the algorithm in terms of asynchronous rounds, where a round means the smallest time interval in which each robot, which has not yet finished the algorithm, has been activated at least once. We show that this algorithm needs $O(n^2)$ asynchronous rounds. Our analysis provides the first asymptotic upper bound on the running time in terms of asynchronous rounds.
  
Then we show how the number of colors can be reduced to $O(1)$ at the cost of the running time. The algorithm with $1$ hop visibility range, $O(\log \Delta)$ bits of persistent memory, and $O(1)$ colors needs $O(n^2\log \Delta)$ rounds. We show how the running time can be improved by robots with a visibility range of $2$ hops, $O(\log \Delta)$ bits of persistent memory, and $\Delta + 4$ colors (including the color when the light is off). We show that the algorithm needs $O(n)$ asynchronous rounds. Finally, we show how to extend our solution to the $k$-Door case, $k\geq 2$, by using $\Delta + k + 4$ colors, including the color when the light is off. 
\end{abstract} 
 
\section{Introduction}
In swarm robotics, a large number of autonomous mobile robots cooperate to achieve a complex goal. 
The robots of the swarm are simple, cheap, and computationally limited.
They act according to local rules of  behavior.
Robot swarms can achieve high scalability, fault tolerance, and cost-efficiency.

The robots can cooperatively solve different problems, as gathering, flocking, pattern formation, dispersing, filling, coverage, and exploration (e.g. \cite{Albers2000,Amir19,Augustine2018,Barrameda2008,Barrameda2014,Cohen2005,Scheideler19,Hsiang2002,LukovszkiH14}.

The Filling (or, Uniform Dispersal) problem was introduced by Hsiang et al. \cite{Hsiang2002}, where the robots enter an a priori unknown but connected area and have to disperse.
The area is subdivided into pixels, and at the end of the dispersion, each pixel has to be occupied by exactly one robot.

\vspace{1ex}
\noindent\textbf{Model:}
We consider the Filling problem, where the area is represented by a connected graph, which is unknown for the robots.
The robots enter the graph one-by-one through a specific vertex, which is called the \textit{Door} and have to disperse
to cover all vertices of the graph while avoiding collision, i.e. two or more robots can not be at the same vertex.
At the end of the dispersion, each vertex of the graph has to be occupied by exactly one robot.
When the Door vertex becomes empty, a new robot is placed there immediately.
We assume that the robots know an upper bound $\Delta$ on the maximum degree of the graph.

For simplicity, we assume that the degree of the Door vertex is $1$.
Otherwise, we introduce an auxiliary vertex of degree $1$ connected only to the Door, which takes the role of the original Door (this models the two sides of a doorstep).
We assume that, for each vertex $v$, the adjacent vertices are arranged in a fixed cyclic order. This cyclic order is only visible for robots at $v$, and it does not change during the dispersion. 
When a robot $r$ arrives at vertex $v$ from a vertex $u$, then the cyclic order of neighbors is used by $r$ as a linear order of $\deg(v)-1$ neighbors by cutting and removing $u$.

The robots act according to the Look-Compute-Move (LCM) model.
In this model, their actions are decomposed into three phases: In the \textit{Look} phase, a robot takes a snapshot of its surroundings, i.e. the vertices and the robots within its visibility range. In the \textit{Compute} phase, it performs calculations based on the surrounding and determines a neighboring vertex as target vertex, or decide to stay at place. In the \textit{Move} phase, if necessary, it moves to the target vertex.  
The next LCM cycle starts when the target is reached.

Based on the activation times of the robots, there are three main synchronization models studied in the literature: the fully synchronous (FSYNC), the semi-synchronous (SSYNC), and the asynchronous (ASYNC).
In the FSYNC model, all robots are activated at the same time, and they perform their Look, Compute, and Move phases synchronously at the same time, which is ensured by a global clock.
In the SSYNC model, some robots might skip an LCM cycle and stay inactive.
In the ASYNC model, there is no common notion of time available: the robots activate independently after a finite but arbitrary long time, and perform their LCM cycles.
Moreover, their LCM cycle length is not fixed; it also can be arbitrarily long.

The robots are \textit{autonomous}, i.e. no central coordination is present, \textit{homogeneous}, i.e. all the robots have the same capabilities and behaviors, \textit{anonymous}, i.e. they cannot distinguish each other, \textit{myopic}, i.e. they have limited visibility range, and \textit{silent}, i.e. they have no communication capabilities and cannot directly talk to one another.
However, \textit{luminous} robots can communicate indirectly by using a light.
Such robots have a light attached to them, which is externally visible by every robot in their visibility range.
They can use a finite set of colors (including the color when the light is off) 
representing the value of a state variable.
The robots are allowed to change these colors in their Compute phase.
We denote the availability of lights using a superscript representing the number of colors.
In particular, we denote by $X^i$ the model $X \in \{$ASYNC, SSYNC, FSYNC$\}$ when every robot is enhanced
by a light with $i > 1$ colors. In the ASYNC$^{O(1)}$ model, the robots use a constant number of colors (see, e.g. \cite{das2012}).

\vspace{1ex}
\noindent\textbf{Related Work:}
The Filling (or, Uniform Dispersal) problem was introduced by Hsiang et al. \cite{Hsiang2002}, where the robots enter an unknown but connected orthogonal area and have to disperse.
The area is subdivided into pixels, and at the end of the dispersion, each pixel has to be occupied by exactly one robot.
Hsiang et al. \cite{Hsiang2002} considered this problem in the FSYNC model. They assumed that robots have a limited ability to communicate with nearby robots. They proposed two solutions, BFLF and DFLF, both modeling generally known algorithms: BFS and DFS. DFLF required a visibility range of 2 hops. It was assumed that the robots are able to detect the orientation of each other.
Barrameda et al. \cite{Barrameda2008,Barrameda2014} investigated the asynchronous case. In \cite{Barrameda2008} the authors assumed common top-down and left-right directions for the robots and showed that robots with visibility range of 1 hop and 2 bits of persistent memory could solve the problem in an orthogonal area if the area does
not contain holes, without using explicit communication in finite time.
In \cite{Barrameda2014} Barrameda et al. presented two methods for filling an unknown orthogonal area in presence of obstacles (holes) in the ASYNC model. Their first method, called TALK, requires a visibility range of 2 hops\footnote{In \cite{Barrameda2014} it is assumed that the robot sees all eight sourrounding cells and able to communicate with robots at that eight cells. Assuming orthogonal movements, a cell sharing only one corner with the current cell of the robot are reachable in two hops.} if the robots have explicit communication.
The other method, called MUTE, does not use explicit communication between the robots, but it requires a visibility range of 6. Both methods need $O(1)$ bits of persistent memory and terminate in finite time.

In \cite{hideg2017,hideg2018} 
the Filling problem has been investigated in the FSYNC model.
In \cite{hideg2017} %
the authors gave a solution for the orthogonal Filling problem by using robots with $1$ hop visibility range and $O(1)$ bits of persistent memory for both the Single and Multiple Door cases.
In \cite{hideg2018} %
a method for a general Filling problem has been presented, where the area is represented by an arbitrary connected graph.
The robots require $1$ hop visibility range and $O(\Delta)$ bits of persistent memory, where $\Delta$ is the degree of the graph.
For the $k$-Door case, the memory requirement is $O(\Delta \cdot \log k)$.
The general method is called the Virtual Chain Method (VCM), which is a leader-follower method.
In the VCM, the robots form a chain and fill the area mimicking a DFS-like traversal of the graph.
The algorithms presented in \cite{hideg2017} and \cite{hideg2018} are intensively utilizing the synchronous nature of the model to avoid collisions and backtracking.

The model of luminous robots was introduced by Peleg \cite{Peleg2005}. Subsequently, significant amount of research has been carried for a plenty of problems using this model (e.g. 
\cite{Sharma2018ipdps,Bhagat2017,Bose2019,Feletti2018,Flocchini19opodis,Kamei19,DiLuna17,Tixeuil2018,Sharma2016,Sharma2017,Sharma2017ipdps}).
Das et al. \cite{das2012,Flocchini16} considered the model, where the robots can move in the continuous Euclidean plane, and they proved that the asynchronous model with a constant number of colors ASYNC$^{O(1)}$ is strictly more powerful than the semi-synchronous model SSYNC, i.e. ASYNC$^{O(1)}$ $>$ SSYNC. 
Das et al. \cite{Flocchini16} also prove that there are problems that robots cannot solve without lights, even if they are fully synchronous, but can be solved by asynchronous luminous robots with $O(1)$ colors.

D'Emidio et al. \cite{Navarra16} have shown that on graphs one task can be solved in the fully synchronous model FSYNC but not in the asynchronous lights-enhanced model, while for other tasks, the converse holds.
In this work, we show that the Filling problem can be solved in both models by robots with $1$ hop visibility range and $O(1)$ bits of persistent memory.

\noindent\textbf{Our Contribution:}
In this work, we present solutions for the Filling problem by luminous robots on graphs in the ASYNC$^{O(1)}$ model.

First, we describe a method, called PACK, which solves the problem by robots with $1$ hop visibility range, $O(\log \Delta)$ bits of persistent memory, and $\Delta + 4$ colors for the single Door case, including the color when the light is off.
We analyze the algorithm in terms of asynchronous rounds, where a round means the smallest time interval in which each robot, which has not yet finished the algorithm, has been activated at least once. We show that this algorithm needs $O(n^2)$ asynchronous rounds. 
Regarding asynchronous algorithms for the Filling problem, former works only guarantee termination within finite time. Our analysis provides the first asymptotic upper bound on the running time in terms of asynchronous rounds.

Then we show how the number of colors can be reduced to $O(1)$ at the cost of running time. The algorithm with $1$ hop visibility range, $O(\log \Delta)$ bits of persistent memory, and $O(1)$ colors needs $O(n^2\log \Delta)$ rounds.

After this, we show how the running time can be significantly improved by robots with a visibility range of $2$ hops, with no communication, $O(\log \Delta)$ bits of persistent memory, and $\Delta + 4$ colors, by presenting the algorithm called BLOCK.
This algorithm needs $O(n)$ rounds.

Then we extend the BLOCK algorithm for solving the $k$-Door Filling problem, $k\geq 2$, by using $O(\log \Delta)$ bits of memory and $\Delta + k + 4$ colors, including the color when the light is off.   
The visibility range of $2$ hops is optimal for the $k$-Door case (a counterexample when this problem cannot be solved in the ASYNC model with a visibility range of $1$ hop was presented in \cite{Barrameda2008}, also holds for the ASYNC$^{O(1)}$ model).

\section{PACK Algorithm}
Now we describe the PACK algorithm to solve the Filling problem for an area represented by a connected graph of $n$ vertices. PACK is based on the Virtual Chain Method described in \cite{hideg2018}, in which the robots filled the area in a DFS-like dispersion.

The robots are allowed to be in one of the following states: None, Follower, Leader, Finished.
They are initialized with None state when placed at the Door.
The first robot becomes the Leader and moves to a vertex that has never been occupied before (these vertices are called \textit{unvisited} vertices).
The rest of the robots will become Followers and follow the Leader, until the Leader becomes stuck (i.e. no unvisited neighbors available).
Then the robot behind the Leader, called the \textit{successor} robot, becomes a new Leader and moves if possible.
The previous Leader switches to Finished state. %
The algorithm terminates when each robot is in Finished state.

The name Virtual Chain comes from the fact that all active robots (i.e. the Leader and the Followers) are on the path traversed by the current Leader from the Door. This path is called the chain.
The chain contains only visited vertices, which can be occupied by the Followers.
Each Follower follows its \textit{predecessor}, which is the previously placed robot. %
The difficulty is to select the next target vertex for the Leader with a visibility range of 1 hop by ensuring that no other robot can  move to that vertex because the Leader can not see all adjacent vertices of the target and robots on those vertices do not see the Leader.

We define the state \textit{Packed} for the chain. The chain is in Packed state, when each Follower is immediately behind its predecessor, i.e. each vertex on the path traversed by the current Leader from the Door is occupied by a Follower robot.
In this state none of the robots can move except the Leader.
Therefore, only the Leader has to know this state.

\vspace{1ex}
\noindent\textbf{The concept:}
The Leader moves to unvisited vertices until there is no such neighboring vertex.
Before each movement, the Leader waits for Packed state; thus, it cannot collide with other robots, and the Leader can decide which vertex is unvisited.
When the Leader has no neighboring unvisited vertex, it switches to Finished state and does not move anymore.
Its successor then becomes the Leader and the new Leader moves to other unvisited vertices.
The robots use $\Delta + 4$ colors, including the color, when the light is off.
The first $\Delta$ colors show the direction of the target vertex (for each vertex, the adjacent vertices are arranged in a fixed cyclic order), we refer to them as \textit{DIR} colors.
Furthermore, we use two colors, denoted by \textit{CONF} and \textit{CONF2} colors, for confirming that a robot has seen a DIR color of the predecessor, which allows the predecessor to move. For this purpose, the CONF color is sufficient, when the predecessor is a Follower robot. However, when the predecessor robot is the Leader, and it must change the target vertex after the Packed state is reached (details are provided later) or the predecessor becomes the Leader and it chooses an unvisited target vertex, it indicates the new direction with a new DIR color. Then the CONF2 color is needed for ensuring that the successor has seen the lastly shown DIR color.
Furthermore, we use an additional color, called \textit{MOV} color to indicate that a robot is on the way to its target vertex.

Now we describe the rules followed by the robots in different states.

\textit{Leader:} Can only move to an unvisited vertex.
When it wants to move, it shows the direction it wants to go to by setting the corresponding DIR color, and then it waits until its successor gives a confirmation that it can move by setting its CONF color. During the movement, the Leader shows the MOV color.
When its successor sets its CONF color, the chain is in Packed state.
This means each not occupied vertex is also an unvisited vertex (as each vertex in the path of the Leader is occupied by a robot).
If the Leader is still on the Door vertex, therefore, it does not have a successor, it can move without waiting for the CONF color.

\textit{Follower:} Follows its predecessor.
The Follower robot $r$ sets the CONF color if and only if $i$) the predecessor of $r$ is showing its direction, and $ii$) the successor of $r$ -- if exists -- have set its CONF color (i.e. the successor knows in which direction $r$ will move).
This allows the predecessor $r'$ of $r$ to move to its destination knowing: $i$) all the robots behind $r'$ have set CONF color, and $ii$) the robots behind $r'$ will not move until the predecessor of $r$ moved.
When $r'$ is the Leader, the chain is in Packed state.

\textit{None:} The robots are initialized with None state when they are placed at the Door.
If the robot $r$ in None state has no neighboring robot, then $r$ changes its state to Leader, chooses the unique neighboring vertex as target vertex, sets the MOV color, and starts moving there.
Otherwise, if the robot $r$ in None state has one neighboring robot, then $r$ becomes a Follower and sets the neighbor to its predecessor.

There are three special situations where we need the following additional rules:

\textit{Leader target change}: It might happen that the Leader $r$ chooses a target vertex $v$, which is unoccupied at the moment when $r$ performs its Look operation, however, when the successor of $r$ sets the CONF color and $r$ could start to move to $v$, another robot already moved to $v$. 
In such case, the Leader $r$ has to choose a new target, and the successor of $r$ has to know about this choice.
Assume first that $r$ has an unvisited neighboring vertex.
Then $r$ sets the corresponding DIR color and waits until its successor sets the CONF2 color.
Finally, the Leader moves to the target.

Note that the chain is in Packed state when the successor of the Leader $r$ sets the CONF color. 
If the Leader changes the target, no other robot can move until $r$ sets the CONF2 color, and the Leader moves to the target. Consequently, the Leader can change the target vertex only once between two movements. 

If $r$ does not have any unvisited neighboring vertex after $r$ sees the CONF color of the successor, then $r$ can not move anymore and the successor must take the leadership (see the rule below). The robot $r$ sets the $\Delta$ direction color, which has special meaning. The successor $r'$ confirms this by setting the CONF2 color. Then $r$ turn off its light $r'$ becomes the Leader. 
(Note that it would be possible to omit the Leader target change rule by introducing a new color for signaling the Packed state. Then the Leader would only show its direction once the Packed state is achieved, which could be acknowledged with the CONF color.)

\textit{Taking the leadership}: When the Leader $r$ cannot move anymore, its successor has to become the new Leader.
The Leader $r$ indicates that it does not have any unvisited neighboring vertex by setting its direction color to $\Delta$.
I.e. this color has a special meaning: it indicates that the Leader cannot move anymore and wants to switch to Finished state, and the leadership must be taken by its successor.  
When this is detected by the successor $r'$, it sets its CONF color, waits for the previous Leader to turn off its light, 
then $r'$ becomes the Leader. Afterwards, $r'$ tries to move to an unvisited vertex.

\textit{Setting movement color}: 
Before performing the movement, the robots have to set their color to MOV.
Keeping the old color could lead to an error. E.g., consider the following situation. 1. The Leader sets a DIR color. 2: The Follower confirms it by the CONF color. 3: The Leader moves by keeping the DIR color. 4: The Follower shows the corresponding DIR color, receives a CONF, and follows the Leader. 5: The Follower reaches its target, sees the old DIR color of the Leader and sets the CONF color, before the Leader chooses the new target. In order to prevent such situations, the  moving robots set their color to MOV and keep this color until the target is reached, and a new target is determined.
After the movement, the robot sets the previous position as its \textit{Entry} vertex

Pseudocode of the PACK algorithm is provided in the Appendix.

\subsection{Analysis}
\begin{lemma} Leader only moves to unvisited vertices. \label{lemma:m1_unvisited}\end{lemma}
\begin{proof}
An unvisited vertex means no robot has occupied it before.
As the Leader can only move when the chain is in Packed state, %
each vertex not occupied by a robot is an unvisited vertex.
Therefore, each unoccupied vertex, which can be chosen by the Leader, is an unvisited vertex. 
\qed
\end{proof}
\begin{lemma} There can be at most one Leader at any time. \label{lemma:m1_one_leader} \end{lemma}
\begin{proof}
Recall the rule \textit{taking the leadership}.
When a Leader $r$ becomes stuck, $r$ signals this with a special color $\Delta$ and switches to Finished state after the successor $r'$ sets the CONF color. Then $r'$ becomes the new Leader, and acts accordingly.

The first robot placed becomes the Leader, and from that time, each robot can become a Leader after the previous one became Finished.
Therefore, at most one Leader can exist at any time during the dispersion. 
\qed
\end{proof}
\begin{lemma} Robots cannot collide. \label{lemma:m1_collision} \end{lemma}
\begin{proof}
When placed, except for the first robot, each robot has only one neighbor robot.
That will be its predecessor, which the robot will follow during the algorithm.
Before the predecessor moves away, it shows which direction it moves. Therefore, the Follower can always follow it.
As each Follower robot has one predecessor, they cannot collide with each other.

However, the Leader does not follow its predecessor (as it does not have any).
It is required to move to unvisited vertices in order to avoid collisions with Followers (which only move to already visited vertices).
As there is only one Leader and it always moves to unvisited vertices, collisions are not possible. 
\qed
\end{proof}
\begin{lemma} PACK fills the area represented by a connected graph. \label{lemma:m1_fills} \end{lemma}
\begin{proof}
For contradiction, assume the area is not filled when the algorithm terminates.
As the area is connected, there is a vertex $v$ that is not occupied and has a neighboring robot in Finished state.
If $v$ is unvisited, let $r$ be the last neighboring robot of $v$, which became Finished.
However, $r$ cannot switch to Finished state since there is an unvisited vertex neighboring to it.
This contradicts the assumption that $v$ remains unoccupied.

Assume now $v$ is unoccupied, but it has been visited during the algorithm.
Let $t$ be the last time $v$ was occupied by a robot $r$.
After $r$ moves from $v$, its successor will occupy $v$.
This contradicts the assumption that $t$ was the last time of occupation of $v$.
This proves the claim the area is filled when the algorithm terminates. 
\qed
\end{proof}
\begin{theorem}
Algorithm PACK fills an area represented by a connected graph in the ASYNC model by robots having a visibility range of $1$ hop, $O(\log \Delta)$ bits of persistent storage, and $\Delta + 4$ colors, including the color when the light is off.\label{thm:m1}
\end{theorem}
\begin{proof}
As the area is filled (by Lemma~\ref{lemma:m1_fills}), and collisions are not possible (by Lemma~\ref{lemma:m1_collision}), the area will be filled without collisions.
The robots require $O(\log \Delta)$ bits of memory to store the following: \textit{$State$} (4 states: 2 bits), \textit{$Target$} (direction of the target vertex: $\lceil \log \Delta \rceil$ bits), \textit{$NextTarget$} (direction of the vertex, where the robot needs to move after the vertex $Target$ is reached: $\lceil \log \Delta \rceil$ bits).
Regarding the number of colors, the robots use $\Delta$ colors to show the direction where the target of the robot is. There are two additional colors (CONF and CONF2) for confirming the robot saw the signaled direction of the predecessor and one color (MOV) during the movement. 
\qed
\end{proof}
Now we analyze the running time of the algorithm
in terms of asynchronous rounds.
An asynchronous round means the shortest time in which each robot, which is not in Finished state yet, has been activated at least once and performed an LCM cycle. 

\begin{theorem}
The algorithm PACK runs in $O(n^2)$ asynchronous rounds. \label{thm:m1b}
\end{theorem}
\begin{proof}
Assume a chain containing $r_1$, $r_2$, $\dots$, $r_i$ %
(where $r_1$ is the active Leader, and $r_2$, $\dots$, $r_i$ are on the path from the Leader to the Door), and assume that the chain is in Packed state.

Assume first that the Leader $r_1$ has an unoccupied neighboring vertex.
Denote by $T$ the time between two consecutive movements of the Leader.
We divide $T$ into three time intervals: $T=T_1 + T_2 + T_3$.
$T_1$ starts with the movement of the Leader, it includes the time, when all robots in the chain, making one step forwards.
$T_2$ starts with placing a new robot at the Door.
In $T_2$ the robots, starting from the Door, set their CONF color one by one.
This CONF color is 'propagated' to the Leader, meaning that the Packed state is reached.
$T_3$ starts when the Leader recognizes the CONF color of the successor, i.e. after achieving the Packed state.
Then the Leader might find its target occupied by another robot.
In this case, the \textit{Leader target change} rule will be used.

Let $t$ be the first asynchronous round of $T_1$, i.e. in round $t$ the Leader $r_1$ moves to its target and sets its direction color.
At the latest in round $t+1$ the robot $r_2$ detects that $r_1$ left its previous vertex $v$ and in that round $r_2$ moves to $v$ and sets its direction color.
This argument can be repeated to all robots until the last one $r_i$ moves at the latest in round $t+i-1$.
Therefore, $T_1 \leq i$.

Now the second phase $T_2$ starts with the placing of a new robot $r_{i+1}$ at the latest in the round $t+i$.
In that round, the new robot $r_{i+1}$ sets its color to CONF.
The robot $r_i$ sees this color at the latest in the next round $t_{i+1}$ and sets its color to CONF in that round.
Repeating this argument for $r_{i-1},\dots,r_2$, we obtain that $r_2$ sets the CONF color at the latest in round $t_{2i-1}$.
Therefore, $T_2 \leq i$.

Now $T_3$ starts. The Leader $r_1$ recognizes the CONF color of its successor at the latest in round $t_{2i}$.
Then Leader knows that the chain is in Packed state. 
If the target vertex $v$ of the Leader is unoccupied, the Leader can move immediately, since in Packed state each unoccupied vertex is unvisited. Otherwise, if $v$ is occupied, the \textit{Leader target change} protocol is performed, i.e. 1: the Leader chooses a new unoccupied neighboring vertex and shows the corresponding DIR color (in round $t_{2i}$ at the latest), 2: then its successor sets its color to CONF2 (in round $t_{2i+1}$ at the latest). At the latest in the round $t_{2i+2}$ the Leader recognizes this and can move.
Then $T = T_1 + T_2 + T_3 \leq 2 i + 2\leq 2n$ rounds.

Assume now that the Leader $r_1$ has no unoccupied neighboring vertex.
If $r_1$ is at the Door, then it turns the light off and switches to Finished state; the graph is filled.
Otherwise, $r_i$ sees the CONF color of $r_2$ and recognizes the Packed state in the round $t_{2i}$ at the latest. 
Then $r_1$ sets its $\Delta$ color in that round. 
The robot $r_2$ recognizes it in round $t_{2i+1}$ at the latest and sets its CONF color. 
The robot $r_1$ sees the CONF color in round $t_{2i+2}$ at the latest, $r_1$ turns its light off and switches to Finished state. 
The robot $r_2$ sees it in round $t_{2i+3}$ at the latest, $r_2$ becomes the new Leader in that round, and checks if there is an unoccupied neighboring vertex. If so, $r_2$ sets the corresponding DIR color in the same round
and waits for the CONF2 color from the successor. The successor sets the CONF2 color in round $t_{2i+4}$ at the latest, and $r_2$ sees the CONF2 color at the latest in round $t_{2i+5}$. 
Since the chain is already in Packed state, the robot $r_2$ can move in the same round (in round $t_{2i+5}$ at the latest). Otherwise, if $r_2$ has no unoccupied neighboring vertex, then the leadership has to be taken by its successor when the successor exists, i.e. $r_2$ is not at the Door. If $r_2$ is at the Door, then $r_2$ turns the light off and switches to Finished state; the graph is filled.

When a Leader can move, it occupies an unvisited vertex within $2n$ asynchronous rounds.
Otherwise, its successor takes the leadership and performs a target change. 
Taking the leadership and the target change need at most $5$ rounds.
Since the leadership is taken at most once by each robot during the whole algorithm, and there are $n$ robots in the filled graph, at most $5n$ time is used for all leadership taking with target change altogether. Therefore, after at most $2n^2+5n=O(n^2)$ rounds all vertices of the graph become filled.
\qed
\end{proof}
\noindent\textbf{Remark}:
In the ASYNC model, a robot can be inactive between two LCM cycles.
Since the inactive phase allowed to be finite but arbitrarily long, an asynchronous round and the runtime of the algorithm can also be arbitrarily long.
In the case where we do not allow inactive intervals between the LCM cycles and every LCM cycle of every robot takes at most $t_{max}$ time then we can upper bound the time of an asynchronous round by $2\cdot t_{max}$.

\begin{corollary}
$(i)$ 
Assume that every LCM cycle of every robot takes at most $t_{max}$ time and there are no inactive intervals between the LCM cycles. 
Then the running time of the PACK algorithm is $O(n^2 t_{max})$.
$(ii)$ 
In the FSYNC model the PACK algorithm needs $O(n^2)$ LCM cycles. 
\end{corollary}

\subsection{Filling of graphs using constant number of colors}
The PACK algorithm uses $\Delta+4$ colors (including the color when the light is off).
We can reduce the number of colors to $O(1)$ at the cost of the running time, as follows.
We encode the $L=\Delta+4$ colors by a sequence of $\lceil \log L \rceil$ bits and transmit this sequence by emulating the Alternating Bit Protocol (ABP), also referred to as Stop-and-wait ARQ (see, e.g. \cite{Tannenbaum10}). This protocol uses a sequence number from $\{0,1\}$ alternately to transmit the bits. The sender has four states corresponding to the transmitted bit $b\in \{0,1\}$ and the sequence number. The receiver has two states that represent which sequence number is awaited. The data bits are accepted with alternating sequence numbers.
This protocol ensures the correct transmission of the bit sequence without duplicates. 

We emulate the ABP by using six different colors, one for each of the four states of the sender and one for each of the two states of the receiver.
Seeing the current color of the sender, the receiver can decode the sequence number and the data bit. When a color corresponding to the correct sequence number is seen, the receiver sets its color, indicating that it waits for the next bit. 
When the sender sees the changed color of the receiver, it sets its color corresponding to the next data bit and the next sequence number.
Therefore, encoding an original color in a sequence of $\lceil \log L \rceil=O(\log \Delta)$ bits and transmitting this sequence
takes $O(\log \Delta)$ rounds.
This leads to the following Theorem.
\begin{theorem}
The modified Algorithm PACK fills an area represented by a connected graph in the ASYNC model by robots having a visibility range of $1$, $O(\log \Delta)$ bits of persistent storage and $O(1)$ colors. The algorithm needs $O(n^2\log \Delta)$ asynchronous rounds.
\end{theorem}

\section{BLOCK Algorithm}
The PACK algorithm solves the Filling problem in arbitrary connected graphs by robots with a visibility range of 1 hop.
An important property of the PACK algorithm is that the Leader can only move when the chain has reached the Packed state.
Now we consider robots with a visibility range of $2$ hops.
Then the robots see each robot, that potentially could choose the same target vertex.
The idea is that the Leader only chooses a vertex $v$ as the target, if the $1$ hop neighborhood of $v$ does not contain any other robot with the light turned on, except when the light showing direction $\Delta$ (i.e. the robot will not move anymore, it wants to switch to Finished state, and waiting for the confirmation of the successor). 
A vertex neighboring to a robot with its light on (except the color $\Delta$) is considered as \textit{blocked} vertex for the Leader.

We introduce the following additional rules for the robots:

\textit{Leader}: The Leader must not choose a blocked vertex as the target.
As the visibility range of the robots is 2 hops, the Leader can identify the blocked neighbors.
When only blocked or occupied vertices surround the Leader, it chooses to terminate its actions (sets the color $\Delta$ and after the confirmation of the successor it switches to Finished state), and the leadership will be taken by its successor.

\textit{Follower}:
Follower robots 'block' all their unoccupied neighboring vertices. As a result, all unoccupied vertices that are part of the chain are blocked: Before a Follower $r$ would move from a vertex $v$, it sets the DIR color corresponding to the target and blocks all of its unoccupied neighboring vertices. In particular, it blocks the target vertex. Thus the Leader cannot choose the same target. Then $r$ waits until the successor $r'$ sets its CONF color and $r$ moves from $v$. During the movement, the MOV color is set, which keeps the same unoccupied vertices blocked.  When $r$ leaves $v$, the vertex $v$ is blocked by $r'$.
These rules ensure that each vertex on the chain is either occupied or blocked.
Consequently, the Leader only moves to unvisited vertices.
The pseudocode of the BLOCK algorithm is provided in the Appendix.
\subsection{Analysis}
\begin{lemma} Leader only moves to unvisited vertices. \label{lemma:m2_unvisited}\end{lemma}
\begin{proof}
Consider a visited vertex $v$ neighboring to the Leader. Let $r$ be the last robot, that occupied $v$.
When $r$ left $v$, the successor of $r$ blocks $v$. Thus, the Leader cannot move to $v$. 
\qed
\end{proof}
\begin{lemma} There can be at most one Leader at a time. \label{lemma:m2_one_leader} \end{lemma}
\begin{proof}
The arguments of Lemma \ref{lemma:m1_one_leader} can be repeated as the rule for taking of the leadership did not change. 
\qed
\end{proof}
\begin{lemma} Robots cannot collide. \label{lemma:m2_collision} \end{lemma}
\begin{proof}
The arguments of Lemma \ref{lemma:m1_collision} can be repeated as Lemma \ref{lemma:m2_one_leader} only allows one Leader, which only can move to unvisited vertices (Lemma~\ref{lemma:m2_unvisited}). 
\qed
\end{proof}
\begin{lemma} BLOCK fills the area represented by a connected graph. \label{lemma:m2_fills}  \end{lemma}
\begin{proof}
We use similar arguments to those in the proof of Lemma~\ref{lemma:m1_fills}.
Assume that all the robots are in Finished state, and there is an unoccupied vertex $v$, such that $v$ has at least one occupied neighbor.
Additionally, to the cases considered in the proof of Lemma~\ref{lemma:m1_fills}, we have to consider the case when $v$ is blocked, and all neighboring robots become Finished.
Let $t$ be last time when a robot $r$, neighboring to $v$, switches to Finished state.
Since all other neighboring robots of $v$ are in Finished state at time $t$, they do not block $v$.
Therefore, at time $t$ the robot $r$ can move to $v$ instead of switching to Finished.
Thus, we have a contradiction. 
\qed
\end{proof}
\begin{theorem}
Algorithm BLOCK fills the area represented by a connected graph in the ASYNC model by robots having a visibility range of $2$ hops, $O(\log \Delta)$ bits of persistent storage, and using $\Delta + 4$ colors, including the color when the light is off.
\end{theorem}
\begin{proof}
We can use the arguments of the proof of Theorem~\ref{thm:m1} as the area is filled (by Lemma~\ref{lemma:m2_fills}), and collisions are not possible (by Lemma~\ref{lemma:m2_collision}), the area will be filled without collisions.
The robots store the same data in their persistent storage as in Theorem~\ref{thm:m1} and use the same set of colors. 
\qed
\end{proof}
Now we provide runtime analysis of the BLOCK algorithm in the fully synchronous model.
\begin{theorem}
In the ASYNC model, the BLOCK algorithm fills the area represented by a connected graph in $O(n)$ asynchronous rounds. \label{thm:m2r}
\end{theorem}
\begin{proof}
Assume that the chain contains the robots $r_1$, $r_2$, $\dots$, $r_j$, where $r_1$ is the current Leader  and $r_2$, $\dots$, $r_j$ are on the path from the Leader to the Door, and assume that the Leader $r_1$ occupied its position and its successor $r_2$ has arrived at the  previous position of $r_1$.
When the first robot $r_1$ is placed at the Door (i.e. $j=1$), it detects in the first asynchronous round, whether it is a Leader or a Follower.
If the only neighbor is unoccupied, it becomes a Leader and moves in the first round. The first round ends.
After $r_1$ left the Door, the next robot is placed there. 

Assume now that $j \geq 2$. Let $r_i$, $i<j$ be a robot ($r_i$ is either a Leader or a Follower) and assume that its successor $r_{i+1}$ at its previous vertex.
Let $t$ be the current asynchronous round. 
If $r_i$ is Leader, i.e. $i=1$, we additionally assume that it has an unblocked and unoccupied neighboring vertex $v$.
Otherwise, if $r_i$ is not the Leader, assume that the target vertex of $v$ of $r_i$ is unoccupied, i.e. the predecessor $r_{i-1}$ left $v$ already. 
Then $r_i$ sets the corresponding DIR color in round $t$.
At the latest in round $t+1$ the robot $r_{i+1}$ sees the DIR color and sets its color to the CONF, allowing $r_i$ to move.
At the latest in round $t+2$ the robot $r_{i}$ detects this and moves to its target $v$ and at the end of that round $r_{i}$ and $r_{i-1}$ become neighbors again.
Then, at the latest in round $t+3$ the robot $r_{i+1}$ detects that $r_i$ left the neighboring vertex $v'$. 
If $r_{i+1}$ is at the Door it moves at the latest in round $t+3$.
Otherwise, if $r_{i+1}$ is not at the Door, and therefore, a successor robot $r_{i+2}$ exists, $r_{i+1}$ has to wait for the confirmation of $r_{i+2}$ before the movement. 
We will show that $r_{i+2}$ must be at the neighbor vertex behind $r_{i+1}$ in round $t+4$ at the latest.
At the latest in round $t+4$ the robot $r_{i+2}$ sets its CONF color.
Therefore, $r_{i+1}$ can move to $v'$ at the latest in round $t+5$ and at the end of that round $r_{i+1}$ and $r_i$ become neighbors again. At the latest in round $t+6$ we have the same situation regarding $r_i$ and $r_i+1$ as in round $t$, i.e. $r_{i-1}$ shows its
DIR color and $r_{i}$ confirms it.

It remains to show that in round $t+4$ at the latest $r_{i+2}$ must be on the neighbor vertex of $r_{i+1}$. If $r_{i+2}$ is at the Door, then it appeared there after $r_{i+1}$ left the Door, i.e. before round $t$ and $r_{i+1}$ did not move; therefore they must be neighbors in round $t+4$.
Otherwise, let $t'$ be the latest round before $t$, where $r_{i+1}$ and $r_{i+2}$ were neighbors and the robot $r_{i+1}$ detects that its predecessor $r_i$ moved from the neighboring vertex and $r_{i+1}$ sets the DIR color. Then we can repeat the arguments with robots $r_{i+1}$, $r_{i+2}$, and round $t'$ described above, and we obtain that 
($i$) $r_{i+1}$ moves in round $t''\leq t'+2$ and at the end of round $t'+2$ the robots $r_{i}$ and $r_{i+1}$ become neighbors again, and
($ii$) $r_{i+2}$ can move again at the latest in round $t''+3$ and in round $t''+4$ at the latest $r_{i+2}$ and $r_{i+1}$ must be  neighbors. Since $t''\leq t$, in round $t+4$ at the latest $r_{i+2}$ and $r_{i+1}$ must be neighbors. 

Summarizing the above description, %
the robot $r_i$ moves at the latest in every $6^{th}$ round if $r_i$ is a Follower or it is a Leader with an unblocked and unoccupied neighbor.

Assume now that $r_i$ is Leader, its successor $r_{i+1}$ is at its previous vertex, and all neighboring vertices of $r_i$ are blocked or occupied in round $t$. Then $r_i$ sets its $\Delta$ color to show the successor that it has to switch to Finished state.
The successor $r_{i+1}$ confirms it at the latest in round $t+1$. At the latest in round $t+2$ the robot $r_i$ becomes Finished and turns the light off. At the latest in round $t+3$ the robot $r_{i+1}$ becomes the new Leader.
Therefore, the leadership is taken within $4$ rounds. At the latest in round $t+3$ the new Leader $r_{i+1}$ shows its new target if there is an unblocked and unoccupied neighboring vertex, or it sets the $\Delta$ to show the successor that it has to switch to Finished state.

When a Leader can move, it occupies an unvisited unblocked vertex in every $6^{th}$ round. Otherwise, its successor takes the leadership.
Since the leadership is taken at most once by each robot during the whole algorithm, and there are $n$ robots in the filled graph, at most $4n$ rounds used for all 'leadership taking'.
Therefore, after $6n+4n=10n$ rounds, all vertices of the graph become filled.
\qed
\end{proof}

\section{Multiple Doors}

We now consider the case in which there are $k\geq 2$ Doors.
For the $k$-Door Filling, there is a situation that cannot be solved by the above methods:
Let $v$ be an unvisited vertex, which is neighboring to (at least) two Leaders $r_1$ and $r_2$.
In order to fill the graph, exactly one of the Leaders, $r_1$ or $r_2$, has to move to vertex $v$.
If one of the robots, say $r_1$, has been activated earlier, then $r_1$ sets the direction color corresponding to $v$, and it prevents $r_2$ to move to $v$ ($r_1$ blocks $v$ from $r_2$).
However, if the activation times of $r_1$ and $r_2$ are exactly the same, then they would set the direction color at the same time, meaning they mutually block each other from moving to $v$.
If $r_1$ or $r_2$ has no other unvisited vertex in their neighborhood, then none of them could move, and particularly, none of them would occupy $v$.

We propose a protocol, which uses a strict priority order between the Leaders originating from different Doors.
We use the concept from \cite{Barrameda2008} and assume that robots entering from different doors have distinct colors.

\textbf{Priority protocol}: The robots have $k$ additional different colors corresponding to the Door they used for entering the area, where $k$ is the number of Doors.
We define a strict total order between these colors, called priority order.
We call these $k$ colors priority colors.
After showing the direction to the successor and after the successor has confirmed it, the Leader sets its color to its priority color (instead of the MOV color) and starts its movement. It arrives to its target showing its priority color.
We modify the blocking rule for the Leader as follows: If there is a robot with a direction color (except the special color $\Delta$), or confirmation color, or MOV color, or priority color with higher priority than $r$, then its neighbors are considered as blocked.
Since there is a strict total order between the priority colors, in such a situation exactly one of them is allowed to move there.

We modify the rule \textit{taking the leadership}: when the successor robot $r$ notices that the Leader is switching to Finished state (by setting the direction color to $\Delta$), $r$ confirms it by setting its color to the priority color of the old Leader.
\begin{lemma} Priority protocol does not allow collisions. \label{lemma:m3_collision} \end{lemma}
\begin{proof}
Assume $v$ is an unvisited vertex that is neighboring to two Leaders $r_1$ and $r_2$, i.e. both $r_1$ and $r_2$ could move to $v$. %
Let $t_1$ (resp. $t_2$) be the first activation time of $r_1$ (resp. $r_2$) after it has arrived at its current position.

If $t_1 \neq t_2$, then one of them, which was activated before the other one, will block the other one, and they cannot collide.
If $t_1 = t_2$, then they will see each other's priority color.
Then they can decide which robot has a higher priority.
The robot with higher priority will block the other one.
Consequently, Leaders cannot collide with each other.

Now we show that the collision with a Follower is also not possible.
When the Follower $r$ would move to a vertex $v$ it has its CONF color set allowing the predecessor $r'$ to leave $v$.
This blocks $v$ for all Leaders. The predecessor $r'$ also blocks $v$ until $r$ occupies it. Therefore, $r$ cannot collide with a Leader.
\qed
\end{proof}
\begin{lemma} The BLOCK algorithm extended with the Priority protocol fills the connected graph. \label{lemma:m3_fills} \end{lemma}
\begin{proof}
We can repeat the arguments of the proof of Lemma~\ref{lemma:m2_fills}. 
\qed
\end{proof}
\begin{theorem} Algorithm BLOCK extended with the Priority protocol solves the $k$-Door Filling problem, $k\geq 2$, in the ASYNC model in finite time, with $2$ hops of visibility, $O(\log \Delta)$ bits of memory and using $\Delta + k + 4$ colors including the color when the light is off. \end{theorem}
\begin{proof}
We can use the arguments of the proof of Theorem~\ref{thm:m1} as the area is filled (by Lemma~\ref{lemma:m3_fills}), and collisions are not possible (by Lemma~\ref{lemma:m3_collision}), the area will be filled without collisions.
The robots store the same data in their persistent storage as in Theorem~\ref{thm:m1} and use $\Delta + k + 4$ colors. 
\qed
\end{proof}

\section{Summary}
\label{sec:summary}

In this work, we have presented solutions for the Filling problem by luminous robots in the ASYNC$^{O(1)}$ model.
We have presented a method, called PACK, which solves the problem by robots with $1$ hop visibility range, $O(\log \Delta)$ bits of persistent memory, and $\Delta + 4$ colors for the single Door case, including the color when the light is off.
We have shown that this algorithm needs $O(n^2)$ asynchronous rounds. Regarding asynchronous algorithms for the Filling problem, former works only guarantee termination within finite time. Our analysis provides the first asymptotic upper bound on the running time in terms of asynchronous rounds.

Then we have shown how the number of colors can be reduced to $O(1)$ at the cost of running time. The algorithm with $1$ hop visibility range, $O(\log \Delta)$ bits of persistent memory, and $O(1)$ colors needs $O(n^2\log \Delta)$ rounds.

After this, we have shown how the running time can be significantly improved by robots with a visibility range of $2$ hops, with no communication, $O(\log \Delta)$ bits of persistent memory, and $\Delta + 4$ colors, by presenting the algorithm called BLOCK.
This algorithm needs $O(n)$ rounds.

Then we have extended the BLOCK algorithm for solving the $k$-Door Filling problem, $k\geq 2$, by using $O(\log \Delta)$ bits of memory, and $\Delta  + k + 4$ colors, including the color when the light is off.   
The visibility range of $2$ hops is optimal for the $k$-Door case (a counterexample when this problem cannot be solved in the ASYNC model with a visibility range of $1$ hop was presented in \cite{Barrameda2008}, also holds for the ASYNC$^{O(1)}$ model).

\begin{small}
\bibliography{bibliography.bib}
\end{small}
\section*{Appendix}
\appendix
\section{Algorithms PACK and BLOCK}

\begin{breakablealgorithm}
\caption{(\textbf{PACK}): Rules followed by robot $r$.}
\label{alg:pack_algorithm}
\begin{algorithmic}[1]
  \STATE If $r$.$State$ is Follower: \\
   \hspace{\algorithmicindent}If $r$.$NextTarget$ is not set: \\
   \hspace{2\algorithmicindent}If $r$.$Predecessor$ shows DIR color:  \\
   \hspace{ 3\algorithmicindent}Store shown DIR as $r$.$NextTarget$  \\
   \hspace{2\algorithmicindent}Else If $r$.$Predecessor$ shows DIR color $\Delta$:  \\
   \hspace{ 3\algorithmicindent}$r$ switches to Leader state  \\
   \hspace{\algorithmicindent}Else: \\   
   \hspace{2\algorithmicindent}If $r$.$Color$ is not set to CONF: \\
   \hspace{ 3\algorithmicindent}If $r$.$Entry$ is not set: \\
   \hspace{  4\algorithmicindent}Set $r$:$Color$ to CONF \\
   \hspace{ 3\algorithmicindent}Else If $r$.$Entry$ is occupied and $r$.$Successor$ has light set to a CONF color: \\
   \hspace{  4\algorithmicindent}Set $r$:$Color$ to CONF \\
   \hspace{2\algorithmicindent}Else If $r$.$Target$ is unoccupied: \\
   \hspace{ 3\algorithmicindent} $r$ sets $r$.$Color$ to MOV \\
   \hspace{ 3\algorithmicindent} $r$ moves to $r$.$Target$ and  \\
   \hspace{ 3\algorithmicindent} $r$ sets $r$.$Color$ to match $r$.$NextTarget$ \\
   \hspace{ 3\algorithmicindent} $r$ sets $r$.$Target$ to $r$.$NextTarget$ \\
   \hspace{ 3\algorithmicindent} $r$ clears $r$.$NextTarget$ \\
   \hspace{2\algorithmicindent}Else If $r$.$Predecessor$ shows DIR color $\Delta$:  \\
   \hspace{ 3\algorithmicindent}$r$ switches to Leader state  \\
   \hspace{2\algorithmicindent}Else If $r$.$NextTarget$ has been set: \\
   \hspace{ 3\algorithmicindent}If $r$.$Predecessor$ shows different direction:\\
   \hspace{  4\algorithmicindent}$r$ sets $r$.$NextTarget$ to new shown direction\\
   \hspace{  4\algorithmicindent}$r$ sets light to CONF2 color

  \STATE If $r$.$State$ is Leader: \\
   \hspace{\algorithmicindent}If $r$.$Target$ is not set: \\
   \hspace{2\algorithmicindent}$r$ sets $r$.$Target$ to first empty neighbor \\
   \hspace{2\algorithmicindent}If $r$.$Target$ is not set (no empty neighbor found): \\
   \hspace{ 3\algorithmicindent}$r$ sets DIR color $\Delta$ and becomes Finished \\
   \hspace{\algorithmicindent}Else: \\
   \hspace{2\algorithmicindent}If $r$ is waiting for CONF: \\
   \hspace{ 3\algorithmicindent}If $r$.$Successor$ has light set to a CONF color: \\
   \hspace{  4\algorithmicindent}If $r$.$Target$ is unoccupied: \\
   \hspace{   5\algorithmicindent}$r$ sets $r$.$Color$ to MOV \\
   \hspace{   5\algorithmicindent}$r$ moves to $r$.$Target$ \\
   \hspace{   5\algorithmicindent}$r$ clears $r$.$Target$ \\
   \hspace{  4\algorithmicindent}Else If $r$ has unoccupied neighbor $v$:\\
   \hspace{   5\algorithmicindent}$r$ sets $r$.$Target$ and set DIR color to match the direction of $v$ \\
   \hspace{   5\algorithmicindent}$r$ is now waiting for CONF2 \\
   \hspace{  4\algorithmicindent}Else: \\
   \hspace{   5\algorithmicindent}$r$ sets DIR color $\Delta$ and becomes Finished \\
   \hspace{ 3\algorithmicindent}Else: \\
   \hspace{  4\algorithmicindent}Waits for $r$.$Successors$ to set CONF color \\   
   \hspace{2\algorithmicindent}Else If $r$ is waiting for CONF2: \\
   \hspace{ 3\algorithmicindent}If $r$.$Successor$ has light set to a CONF2 color: \\
   \hspace{  4\algorithmicindent}$r$ sets $r$.$Color$ to MOV \\
   \hspace{  4\algorithmicindent}$r$ moves to $r$.$Target$ \\
   \hspace{  4\algorithmicindent}$r$ clears $r$.$Target$ \\
   \hspace{ 3\algorithmicindent}Else: \\
   \hspace{  4\algorithmicindent}Waits for $r$.$Successors$ to set CONF color \\   
   
  \STATE If $r$.$State$ is None: \\
   \hspace{\algorithmicindent}$r$ sets $r$.$Target$ to neighbor \\
   \hspace{\algorithmicindent}If $r.Target$ does not contain robot:\\
   \hspace{2\algorithmicindent}$r$ becomes the Leader\\
   \hspace{2\algorithmicindent}$r$ sets $r$.$Color$ to MOV \\
   \hspace{2\algorithmicindent}$r$ moves to $r$.$Target$\\
   \hspace{2\algorithmicindent}$r$ clears $r$.$Target$\\
   \hspace{\algorithmicindent}Else:\\
   \hspace{2\algorithmicindent}$r$ becomes a Follower\\
  
 \end{algorithmic}
\end{breakablealgorithm}

\begin{breakablealgorithm}
  \caption{(\textbf{BLOCK}): Rules followed by robot $r$.}
  \label{alg:block_algorithm}
  \begin{algorithmic}[1]
    \STATE If $r$.$State$ is Follower: \\
     \hspace{\algorithmicindent}If $r$.$NextTarget$ is not set: \\
     \hspace{2\algorithmicindent}If $r$.$Predecessor$ shows DIR color:  \\
     \hspace{ 3\algorithmicindent}Store shown DIR as $r$.$NextTarget$  \\
     \hspace{ 3\algorithmicindent}Set $r$.$Color$ to CONF  \\
     \hspace{2\algorithmicindent}If $r$.$Predecessor$ shows DIR color $\Delta$:  \\
     \hspace{ 3\algorithmicindent}$r$ switches to Leader state  \\
     \hspace{\algorithmicindent}Else: \\   
     \hspace{2\algorithmicindent}If $r$.$Color$ is set to CONF or CONF2: \\     
     \hspace{ 3\algorithmicindent}If $r$.$Target$ is unoccupied: \\
     \hspace{  4\algorithmicindent}Set $r$.$Color$ to DIR to match $r$.$Target$ \\
     \hspace{ 3\algorithmicindent}Else: \\
     \hspace{  4\algorithmicindent}If $r$.$Predecessor$ shows DIR color $\Delta$:  \\
     \hspace{   5\algorithmicindent}$r$ switches to Leader state  \\
     \hspace{ 3\algorithmicindent}Else If $r$.$NextTarget$ has been set: \\
     \hspace{  4\algorithmicindent}If $r$.$Predecessor$ shows different direction:\\
     \hspace{   5\algorithmicindent}$r$ sets $r$.$NextTarget$ to new shown direction\\
     \hspace{   5\algorithmicindent}$r$ sets light to CONF2 color\\
     \hspace{2\algorithmicindent}Else: \\
     \hspace{ 3\algorithmicindent}If $r$.$Entry$ is not set or $r$.$Entry$ is occupied \\
     \hspace{ 3\algorithmicindent}and $r$.$Successor$ has light set to a CONF color: \\
     \hspace{  4\algorithmicindent}Else If $r$.$Target$ is unoccupied: \\
     \hspace{   5\algorithmicindent}$r$ sets $r$.$Color$ to MOV \\
     \hspace{   5\algorithmicindent}$r$ moves to $r$.$Target$ \\
     \hspace{   5\algorithmicindent}sets light to match $r$.$NextTarget$ \\
     \hspace{   5\algorithmicindent}$r$ sets $r$.$Target$ to $r$.$NextTarget$ \\     
     \hspace{  4\algorithmicindent}Else If $r$.$Predecessor$ shows DIR color $\Delta$:  \\
     \hspace{   5\algorithmicindent}$r$ switches to Leader state \\
     \hspace{   5\algorithmicindent}$r$ is now waiting for CONF2 \\
  
    \STATE If $r$.$State$ is Leader: \\
     \hspace{\algorithmicindent}If $r$.$Target$ is not set: \\
     \hspace{2\algorithmicindent}$r$ sets $r$.$Target$ to first empty and not blocked neighbor \\
     \hspace{2\algorithmicindent}If $r$.$Target$ is not set (no empty neighbor found): \\
     \hspace{ 3\algorithmicindent} $r$ sets DIR color $\Delta$ and becomes Finished \\
     \hspace{\algorithmicindent}Else: \\
     \hspace{2\algorithmicindent}If $r$ is waiting for CONF: \\
     \hspace{ 3\algorithmicindent}If $r$.$Entry$ is not set: \\
     \hspace{  4\algorithmicindent}$r$ sets $r$.$Target$ to first empty and not blocked neighbor \\
     \hspace{  4\algorithmicindent}If $r$.$Target$ is not set (no empty neighbor found): \\
     \hspace{   5\algorithmicindent}$r$ sets DIR color $\Delta$ and becomes Finished \\
     \hspace{  4\algorithmicindent}Else: \\
     \hspace{   5\algorithmicindent}$r$ sets $r$.$Color$ to MOV \\
     \hspace{   5\algorithmicindent}$r$ moves to $r$.$Target$ \\
     \hspace{   5\algorithmicindent}$r$ clears $r$.$Target$ \\
     \hspace{ 3\algorithmicindent}If $r$.$Entry$ is occupied  \\
     \hspace{  4\algorithmicindent}If $r$.$Successor$ has light set to a CONF or CONF2 color: \\
     \hspace{   5\algorithmicindent}If $r$.$Target$ is occupied: \\
     \hspace{    6\algorithmicindent}$r$ sets $r$.$Target$ to first empty and not blocked neighbor \\
     \hspace{    6\algorithmicindent}If $r$.$Target$ is not set (no empty neighbor found): \\
     \hspace{    6\algorithmicindent} $r$ sets DIR color $\Delta$ and becomes Finished \\
     \hspace{   5\algorithmicindent}Else: \\
     \hspace{    6\algorithmicindent}$r$ sets $r$.$Color$ to MOV \\
     \hspace{    6\algorithmicindent}$r$ moves to $r$.$Target$ \\
     \hspace{    6\algorithmicindent}$r$ clears $r$.$Target$ \\
     \hspace{  4\algorithmicindent}Else: \\
     \hspace{   5\algorithmicindent}$r$ sets $r$.$Target$ to first empty and not blocked neighbor \\
     \hspace{   5\algorithmicindent}If $r$.$Target$ is not set (no empty neighbor found): \\
     \hspace{   5\algorithmicindent} $r$ sets DIR color $\Delta$ and becomes Finished \\
     \hspace{2\algorithmicindent}Else If $r$ is waiting for CONF2: \\
     \hspace{ 3\algorithmicindent}If $r$.$Successor$ has light set to a CONF2 color: \\
     \hspace{  4\algorithmicindent}$r$ sets $r$.$Color$ to MOV \\
     \hspace{  4\algorithmicindent}$r$ moves to $r$.$Target$ \\
     \hspace{  4\algorithmicindent}$r$ clears $r$.$Target$ \\
     \hspace{ 3\algorithmicindent}Else: \\
     \hspace{  4\algorithmicindent}Waits for $r$.$Successors$ to set CONF2 color \\   
          
    \STATE If $r$.$State$ is None: \\
     \hspace{\algorithmicindent}$r$ sets $r$.$Target$ to neighbor \\
     \hspace{\algorithmicindent}If $r.Target$ does not contain robot:\\
     \hspace{2\algorithmicindent}$r$ becomes the Leader\\
     \hspace{2\algorithmicindent}$r$ sets light to MOV\\
     \hspace{2\algorithmicindent}$r$ moves to $r$.$Target$\\
     \hspace{2\algorithmicindent}$r$ clears $r$.$Target$\\
     \hspace{\algorithmicindent}Else:\\
     \hspace{2\algorithmicindent}$r$ becomes a Follower\\    
   \end{algorithmic}
  \end{breakablealgorithm}

\section{Simulation Results}

We have implemented our algorithms and conducted simulations on different graph topologies.  
The tested topologies are Line graphs, Stars, and Delaunay triangulations with vertices uniformly randomly distributed in a square area.  
For the runtime, we assumed the FSYNC model (i.e. all robots are active in every LCM-cycle), therefore the runtimes are better comparable. 

\subsection{Single Door}
\subsubsection{Line graph}

A line graph consist of $n\geq 1$ vertices $V=\{v_1,...,v_n\}$ and edges $E=\{(v_i,v_{i+1}): 1\leq i< n\}$. The Door vertex is at $v_1$ in the end of the line. 
In this case there are no branching vertices, and the robots move on a unique path. 

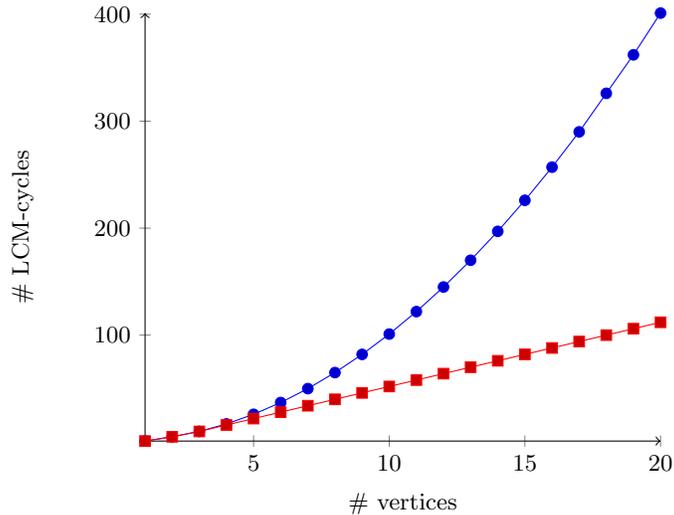
\begin{figure}[!ht]
  \centering  
  \begin{tikzpicture}
      \begin{axis}[axis lines=middle,
        axis line style={->},
        x label style={at={(axis description cs:0.5,-0.1)},anchor=north},
        y label style={at={(axis description cs:-0.2,.5)},rotate=90,anchor=south},
        xlabel={\# vertices},  
        ylabel={\# LCM-cycles}]  
      \addplot table [x=a, y=b, col sep=comma] {line_PACK.csv};
      \addplot table [x=a, y=b, col sep=comma] {line_BLOCK.csv};
      \end{axis} 
  \end{tikzpicture}
  \caption{Simulation results for Line graphs. The $x$-axis represents the number of nodes, $y$-axis represents the number of required LCM-cycles to finish the filling. The runtime of the PACK algorithm is displayed with blue color, and of the BLOCK algorithm with red.} 
\label{fig:line_simulation}
\end{figure}

The line graph exhibits a worst-case input for the PACK algorithm, since one Leader traverse the whole line, and between two consecutive steps of the Leader all robots must form a Packed chain. This results in quadratic running time. 
The BLOCK algorithm runs in linear time, which is also confirmed by the simulations. (Figure \ref{fig:line_simulation}).

\subsubsection{Star graph}

A star graph of $n\geq 1$ consists of one central vertex $v_1$, which is connected to all other vertices $\{v_2,...,v_n\}$ by an edge. All vertices in $\{v_2,...,v_n\}$ are only connected to $v_1$ by an edge. The Door is placed at one of the degree 1 vertices. 
In this topology, the Leader first moves to the central vertex, then to one of the degree 1 nodes, and becomes Finished. The leadership is taken by its Follower occupying $v_1$. Then the new Leader moves to one of the leaves, and the leadership is taken by its Follower occupying $v_1$, etc... In this case, the lenght of the chain behind the current Leader is at most 2, and the Packed state is achieved in a constant number of LCM-cycles. Therefore, the PACK algorithm runs in linear time on the star. 

\begin{figure}[!ht]
  \centering  
  \begin{tikzpicture}
      \begin{axis}[axis lines=middle,
        axis line style={->},
        x label style={at={(axis description cs:0.5,-0.1)},anchor=north},
        y label style={at={(axis description cs:-0.2,.5)},rotate=90,anchor=south},
        xlabel={\# vertices},  
        ylabel={\# LCM-cycles}]  
      \addplot table [x=a, y=b, col sep=comma] {star_PACK.csv};
      \addplot table [x=a, y=b, col sep=comma] {star_BLOCK.csv};
      \end{axis} 
  \end{tikzpicture}
  \caption{Simulation results for Star graphs. The $x$-axis represents the number of nodes, the $y$-axis represents the number of required LCM-cycles to finish the filling. Blue shows the runtime of the PACK algorithm; red shows the runtime of the BLOCK algorithm.} 
\label{fig:star_simulation}
\end{figure}
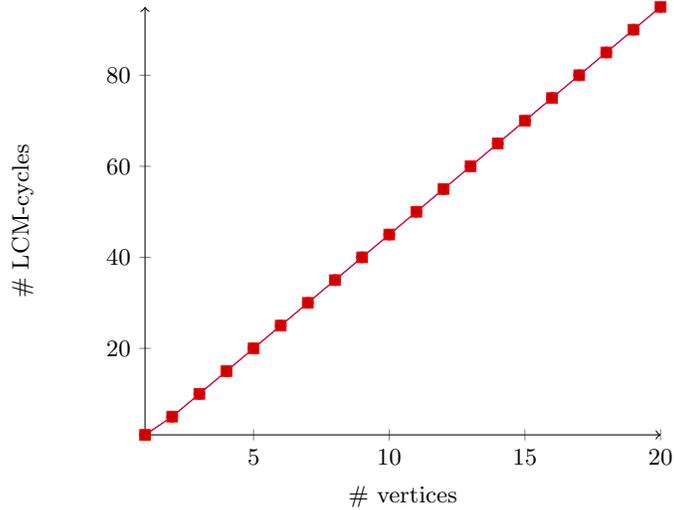

The results can be seen in \figurename~\ref{fig:star_simulation}, which shows that the runtime of the PACK and the BLOCK algorithm is exactly the same in both cases, both runtimes are linear in the number of vertices. 

\subsubsection{Random Delaunay triangulation} 

The graphs are generated by using the following method. 
$i$) In a square area we select $n$ points independently, uniformly at random, where $n$ is the size of the graph. 
$ii$) Using the first $n-1$ points we compute a Delaunay triangulation. 
Then we add $n$-th point as Door vertex as an auxiliary vertex to the closest random vertex.

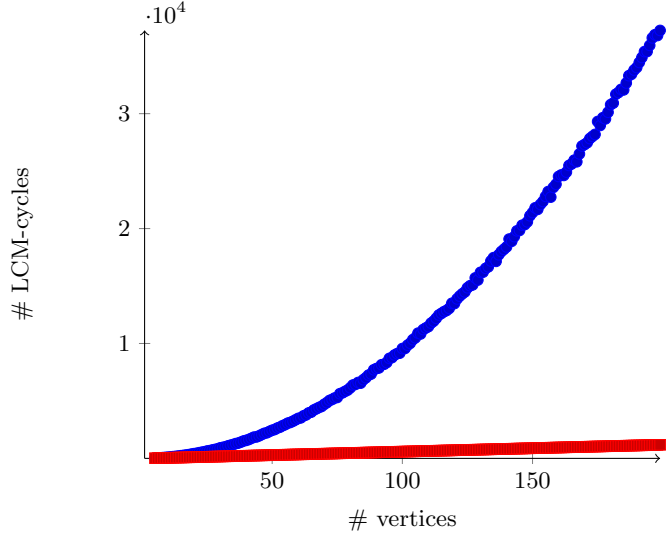
\begin{figure}[!ht]
  \centering  
  \begin{tikzpicture}
      \begin{axis}[axis lines=middle,
        axis line style={->},
        x label style={at={(axis description cs:0.5,-0.1)},anchor=north},
        y label style={at={(axis description cs:-0.2,.5)},rotate=90,anchor=south},
        xmin=1,
        ymin=1,
        xlabel={\# vertices},  
        ylabel={\# LCM-cycles}]  
      \addplot table [x=a, y=b, col sep=comma] {rand_PACK.csv};
      \addplot table [x=a, y=b, col sep=comma] {rand_BLOCK.csv};      
      \end{axis}
  \end{tikzpicture}
  \caption{Simulation results for Delaunay triangulations. %
  The $x$-axis represents the number of vertices, the $y$-axis represents the number of required LCM-cycles to finish the filling. The red curve shows the runtime of the BLOCK algorithm; the blue curve shows the runtime of the PACK algorithm.
  }  
  \label{fig:simulation_del}  
\end{figure}

For this simulation, we generated $50$ random Delaunay graphs using the described method for each vertex set size, $n=3,\dots, 200$. Then, for each input graph, we measured the number of LCM-cycles performed by both the PACK and the BLOCK algorithms. 
Then we computed the average runtimes of the 50 runs of both algorithms for each input size, $n=3,\dots, 200$. 
Figure \ref{fig:simulation_del} shows the simulation results.
The simulations are backing up the linear runtime for the BLOCK algorithm. 
For the PACK algorithm, the simulations suggest quadratic runtime.

\subsection{Multiple Doors}

For this simulation, for $n=1000$ vertices and $k=1,\dots, 200$ Doors, we generate 50 random Delaunay graphs as follows: 
$i$) In a square area, we select $n$ points independently, uniformly at random, where $n$ is the size of the graph.  
$ii$) Using the first $n-k$ points, we compute a Delaunay triangulation.  
Then we add the remaining $k$ points as Door vertices and join each of them with the closest Delaunay vertex. 
The purpose of this simulation was to test the speed-up of the algorithm in case there are multiple Doors. 
For each $k=1,\dots, 200$, we plotted the average runtime on the 50 randomly generated Delaunay triangulations. 
The simulation results in Figure \ref{fig:simulation_md} indicate that runtime of the $k$-Door BLOCK algorithm is proportional to $n/k$ for this simulation setting. 

\begin{figure}[!ht]
    \centering  
    \begin{tikzpicture}
        \begin{axis}[axis lines=middle,
          axis line style={->},
          x label style={at={(axis description cs:0.5,-0.1)},anchor=north},
          y label style={at={(axis description cs:-0.2,.5)},rotate=90,anchor=south},
          ymode=log,
          xmode=log,
          ymin=1,
          xlabel={\# Doors},  
          ylabel={\# LCM-cycles}]  
        \addplot table [x=a, y=b, col sep=comma] {multi_BLOCK.csv};
        \end{axis}
    \end{tikzpicture}
    \caption{Multiple Doors BLOCK algorithm on random Delaunay triangulations with $n=1000$ vertices. The $x$-axis represents the number of Doors; the $y$-axis represents the number of required LCM-cycles to finish the filling. Both axes are log-scaled.} 
    \label{fig:simulation_md} 
  \end{figure}
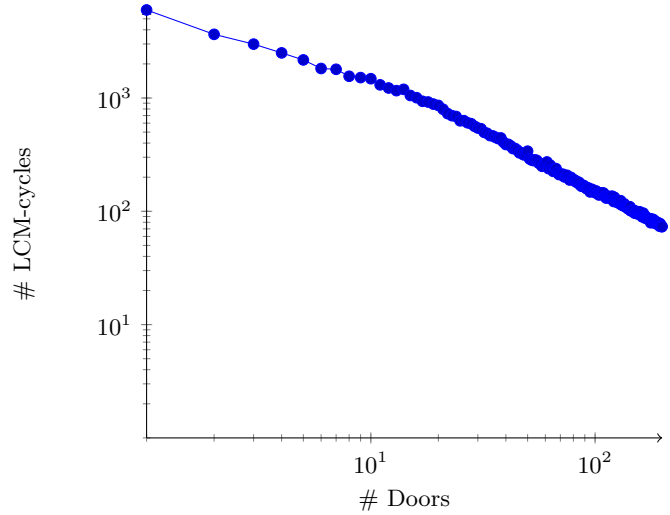

\end{document}